\documentclass[thmsa,12pt]{article}
\usepackage{amssymb}
\usepackage{amsfonts}
\usepackage{amsmath}
\usepackage{amsthm}
\usepackage{latexsym}

\newtheorem{theorem}{Theorem}

\newtheorem{proposition}[theorem]{Proposition}
\newtheorem{corollary}[theorem]{Corollary}

\theoremstyle{definition}

\begin{document}
\title{The Baire partial quasi-metric space: A mathematical tool for asymptotic complexity analysis in Computer Science}

\author{M.A. Cerd\`{a}-Uguet$^{1}$, M.P. Schellekens$^{2}$, O. Valero$^{1}$\thanks{Corresponding author e-mail, telephone and fax numbers: o.valero@uib.es, +34 971 259817, +34 971 173003
 }}

\date{$^{1}$Departamento de Ciencias Matem\'{a}ticas e Inform\'{a}tica, Universidad de las Islas Baleares, 07122, Baleares, Spain \\ e-mails: macerda2@educacio.caib.es, o.valero@uib.es\\
$^{2}$Center of Efficiency-Oriented Languages, Department of Computer Science, University College Cork, Western Road, Cork, Ireland \\ e-mail: m.schellekens@cs.ucc.ie}

\maketitle

\begin{abstract}
In 1994, S.G. Matthews introduced the notion of partial metric space in
order to obtain a suitable mathematical tool for program verification [Ann.
New York Acad. Sci. 728 (1994), 183-197]. He gave an application of this new
structure to parallel computing by means of a partial metric version of the celebrated Banach fixed point theorem [Theoret. Comput. Sci. 151\textbf{%
\ }(1995), 195-205]. \ Later on, M.P. Schellekens introduced the theory of
complexity (quasi-metric) spaces as a part of the development of a
topological foundation for the asymptotic complexity analysis of programs and
algorithms [Electronic Notes in Theoret. Comput. Sci. 1 (1995),
211-232]. The applicability of this theory to the asymptotic complexity analysis of
Divide and Conquer algorithms was also illustrated by Schellekens. In particular,
he gave a new proof, based on the use of the aforenamed Banach fixed point theorem, of
the well-known fact that Mergesort algorithm has optimal asymptotic
average running time of computing.\ In this paper, motivated by the utility of partial
metrics in Computer Science,\ we discuss whether the Matthews fixed point theorem
is a suitable tool to analyze the asymptotic complexity of algorithms in the
spirit of Schellekens. Specifically, we show that a slight modification of the well-known
Baire partial metric on the set of all words over an alphabet constitutes an
appropriate tool to carry out the asymptotic complexity analysis of algorithms via fixed point methods without the need for assuming the convergence condition inherent to the definition of the complexity space in the Shellekens framework. Finally, in order to illustrate and to validate the developed theory we apply our results to analyze the
asymptotic complexity of Quicksort, Mergesort and Largesort algorithms.\medskip

\noindent 2010 AMS Classification: 47H10, 54E50, 54F05, 68Q25, 68W40.\medskip

\noindent Keywords: asymptotic complexity analysis, recurrence equation, quasi-metric, partial metric, Baire partial metric,
Baire partial quasi-metric, running time of computing, fixed point.

\end{abstract}

\section{Introduction and preliminaries \label{Int}}

Throughout this paper the letters $\mathbb{R}^{+},$ $\mathbb{N}$
and $\mathbb{\omega }$ will denote the set of
nonnegative real numbers, the set of positive integer numbers and the set of
nonnegative integer numbers, respectively.

In Computer Science the complexity analysis of an algorithm is based
on determining mathematically the quantity of resources needed by
the algorithm in order to solve the problem for which it has been
designed. A typical resource, playing a central role in complexity
analysis, is the execution time or running time of computing. Since there
are often many algorithms to solve the same problem, one objective
of the complexity analysis is to assess which of them is faster when
large inputs are considered. To this end, it is necessary to compare
their running time of computing. This is usually done by means of
the asymptotic analysis in which the running time of an algorithm is
denoted by a function $T:\mathbb{N} \rightarrow (0,\infty]$ in such a way that $T(n)$ represents the time taken
by the algorithm to solve the problem under consideration when the
input of the algorithm is of size $n.$  Let us denote by $\mathcal{RT}$ the set formed by all functions from $\mathbb{N}$ into $(0,\infty]$. Of course the running time of an algorithm does not only depend on the input size $n$, but it
depends also on the particular input of the size $n$ (and the
distribution of the data). Thus the running time of an algorithm is
different when the algorithm processes certain instances of input
data of the same size $n.$ As a consequence, for the purpose of size-based comparisons, it is necessary to
distinguish three possible behaviours in the complexity analysis of
algorithms. These are the so-called best case, the worst case and
the average case. The best case and the worst case for an input of
size $n$ are defined by the minimum and the maximum running time of
computing over all inputs of the size $n$, respectively. The average
case for an input of size $n$ is defined by the expected value or
average over all inputs of the size $n$.

In general, to determine exactly the function which describes the running
time of computing of an algorithm is an arduous task. However, in most
situations is sufficient to know the running time of computing of an
algorithm in an \textquotedblleft approximate\textquotedblright\ way rather than in
an exact one. For this reason the asymptotic complexity analysis of
algorithms, based on the $\mathcal{O}$-notation, is focused on obtaining the \textquotedblleft
approximate\textquotedblright\ running time of computing of an algorithm.  Indeed, if $%
g\in \mathcal{RT}$ denotes the running time of
computing of an algorithm then the statement $g\in \mathcal{O}(f)$,
where $f\in \mathcal{RT},$ means that there
exists $n_{0}\in \mathbb{N}$ and $c\in \mathbb{R}^{+}$ such that $g(n)\leq cf(n)$ for
all $n\in  \mathbb{N}$ such that $n\geq n_{0}$ ($\leq $ stands for the usual order on $\mathbb{R^{+}}$)$.$ So
the function $f$ gives an asymptotic upper bound of the running time $g$
and, thus, an \textquotedblleft approximate\textquotedblright\ information
of the running time of the algorithm. The set $\mathcal{O}(f)$ is called
the asymptotic complexity class of $f.$ Hence, from an asymptotic complexity
analysis viewpoint, determining the running time of an algorithm consists
of obtaining its asymptotic complexity class. For a fuller treatment of
complexity analysis of algorithms we refer the reader to \cite{Bratley,Cormen}.

In 1995, M.P. Schellekens introduced the theory of complexity spaces as a part of the development of a topological foundation for the
asymptotic complexity analysis of algorithms (\cite{Sch95}). This theory is
based on the notion of quasi-metric space.

Let us recall that, following \cite{Ku}, a quasi-metric on a nonempty set $X$
is a function $d:X\times X \rightarrow \mathbb{R^{+}}$ such that for all $%
x,y,z\in X:$ \medskip

$
\begin{array}{ll}
(i) & d(x,y)=d(y,x)=0\Leftrightarrow x=y; \\
(ii) & d(x,y)\leq d(x,z)+d(z,y).
\end{array}
$
\medskip

A quasi-metric space is a pair $(X,d)$ such that $X$ is a nonempty set and $%
d $ is a quasi-metric on $X.$

Each quasi-metric $d$ on $X$ generates a $T_{0}$-topology $\mathcal{T}(d)$
on $X$ which has as a base the family of open $d$-balls $\{B_{d}(x,%
\varepsilon ):x\in X,$ $\varepsilon >0\}$, where $B_{d}(x,\varepsilon
)=\{y\in X:d(x,y)<\varepsilon \}$ for all $x\in X$ and $\varepsilon >0.$

Given a quasi-metric $d$ on $X$, the function $d^{s}$ defined on
$X\times X$ by $d^{s}(x,y)=\max\left(d(x,y),d(y,x)\right)$ is a
metric on $X$.

A quasi-metric space $(X,d)$ is called bicomplete if the metric space $%
(X,d^{s})$ is complete.

A well-known example of a bicomplete quasi-metric space is the pair $((0,\infty],u_{-1}),$ where
$$u_{-1}(x,y)=\max\left( \frac{1}{y}-\frac{1}{x}, 0\right)$$ for all $x,y\in (0,\infty]$.
% * I removed this sentence: it occurs below again. "Obviously we adopt the convention that $\frac{1}{\infty}=0$."

The quasi-metric space $((0,\infty],u_{-1})$ plays a central role in the Schellekens theory. Indeed, let us recall that the complexity space is the pair
$(\mathcal{C},d_{\mathcal{C}}),$ where
%\begin{equation*}
$$\mathcal{C}=\{f\in \mathcal{RT} :\sum_{n=1}^{\infty
}2^{-n}\frac{1}{f(n)}<\infty \}$$
%\end{equation*}%
and $d_{\mathcal{C}}$ is the quasi-metric on $\mathcal{C}$ defined by
%\begin{equation*}
$$d_{\mathcal{C}}(f,g)=\sum_{n=1}^{\infty }2^{-n}\max\left(\frac{1}{g(n)}-\frac{1}{f(n)},0\right).$$
%\end{equation*}

Obviously we adopt the convention that $\frac{1}{\infty }=0.$

According to \cite{Sch95}, from a complexity analysis point of view, it is
possible to associate a function of $\mathcal{C}$ with each algorithm in
such a way that such a function represents, as a function of the size of the
input data, the running time of computing of the algorithm$.$ Because of this, the elements
of $\mathcal{C}$ are called complexity functions. Moreover, given two
functions $f,g\in \mathcal{C}$, the numerical value $d_{\mathcal{C}}(f,g)$
(the complexity distance from $f$ to $g$) can be interpreted as the relative
progress made in lowering the complexity by replacing any program $P$ with
complexity function $f$ by any program $Q$ with complexity function $g$.
Therefore, if $f\neq g,$ the condition $d_{\mathcal{C}}(f,g)=0$ can be
read as $f$ is \textquotedblleft at least as efficient\textquotedblright\ as $%
g $ on all inputs (i.e. $d_{\mathcal{C}}(f,g)=0\Leftrightarrow f(n)\leq g(n)$
for all $n\in \mathbb{N} $). Thus we can encode the natural order relation on the
set $\mathcal{C},$ induced by the pointwise order $\leq ,$ through the quasi-metric $d_{%
\mathcal{C}}.$ In particular the fact that $d_{\mathcal{C}}(f,g)=0$ implies
that $f\in \mathcal{O}(g).$

The applicability of this theory to the asymptotic complexity analysis of Divide and Conquer algorithms was illustrated by Schellekens in \cite{Sch95}. In
particular, he gave a new proof of the well-known fact that Mergesort algorithm has optimal asymptotic average
running time of computing. To do this he introduced a method, based on a
fixed point theorem for functionals defined on the complexity space into
itself, to analyze the general class of Divide and Conquer algorithms. Let
us recall the aforenamed method.

A Divide and Conquer algorithm solves a problem of size $n$ ($n\in \mathbb{N} $)
splitting it into $a$ subproblems of size $\frac{n}{b},$ for some constants $%
a,b$ with $a,b\in \mathbb{N}$ and $a,b>1,$ and solving them separately by
the same algorithm. After obtaining the solution of the subproblems, the
algorithm combines all subproblem solutions to give a global solution to the
original problem$.$ The recursive structure of a Divide and Conquer algorithm
leads to a recurrence equation for the running time of computing. In many
cases the running time of a Divide and Conquer algorithm is the solution to a
recurrence equation of the form
\begin{equation}
T(n)=\left\{
\begin{array}{ll}
c & \text{\textrm{if }}n=1 \\
aT(\frac{n}{b})+h(n) & \text{\textrm{if }}n\in \omega _{b}%
\end{array}%
\right. ,  \label{typrec}
\end{equation}%
where $\omega _{b}=\{b^{k}:k\in \mathbb{N}\}$, $c\in \mathbb{R^{+}}$ ($c>0$) denotes the complexity on the base case (i.e. the problem
size is small enough and the solution takes constant time), $h(n)$
represents the time taken by the algorithm in order to divide the original
problem into $a$ subproblems and to combine all subproblems solutions into a
unique one ($h\in \mathcal{RT}$ and $0<h(n)<\infty $ for all $n\in \mathbb{N}$).

Notice that for Divide and Conquer algorithms, it is typically sufficient to obtain the complexity on inputs of size $n$ where $n$ ranges over the
set $\omega _{b}$ (\cite{Bratley,Cormen,Sch95})$.$

Typical examples of algorithms whose running time of computing can be obtained by means of the recurrence (\ref{typrec}) are Quicksort (in the best case behaviour) and Mergesort.

In order to compute the running time of computing of a Divide and Conquer
algorithm satisfying the recurrence equation (\ref{typrec}), it is necessary
to show that such a recurrence equation has a unique solution and, later, to
obtain the asymptotic complexity class of such a solution. The method
introduced by Schellekens to show that the equation (\ref{typrec}) has a
unique solution, and to obtain the asymptotic complexity class of the
solution is outlined below:

Denote by $\mathcal{C}_{b,c}$ the subset of $\mathcal{C}$ given
by
\begin{equation*}
\mathcal{C}_{b,c}=\{f\in \mathcal{C}:f(1)=c\text{ and }f(n)=\infty \text{
for all }n\notin \omega _{b}\text{ with }n>1\}.
\end{equation*}

Since the quasi-metric space $(\mathcal{C},d_{\mathcal{C}})$ is bicomplete
(see Theorem 3 and Remark in page 317 of \cite{RomSch}) and the set $%
\mathcal{C}_{b,c}$ is closed in $(\mathcal{C},d_{\mathcal{C}}^{s}),$ we have
that the quasi-metric space $(\mathcal{C}_{b,c},d_{\mathcal{C}}|_{\mathcal{C}%
_{b,c}})$ is bicomplete.

Next we associate a functional $\Phi _{T}:\mathcal{C}_{b,c}\rightarrow
\mathcal{C}_{b,c}$ with the recurrence equation (\ref{typrec}) of a Divide and Conquer
algorithm defined as follows:

\begin{equation}
\Phi _{T}(f)(n)=\left\{
\begin{array}{ll}
c & \text{\textrm{if }}n=1 \\
\infty  & \text{\textrm{if }}n\notin \omega _{b} \textrm{ and } n>1\\
af(\frac{n}{b})+h(n) & \text{\textrm{otherwise}}%
\end{array}%
\right. .  \label{funct}
\end{equation}%
Of course a complexity function in $\mathcal{C}_{b,c}$ is a solution to the
recurrence equation (\ref{typrec}) if and only if it is a fixed point of the
functional $\Phi_{T} .$ It was proved in \cite{Sch95} that
\begin{equation}
d_{\mathcal{C}}|_{\mathcal{C}_{b,c}}(\Phi _{T}(f),\Phi _{T}(g))\leq \frac{1}{%
a}d_{\mathcal{C}}|_{\mathcal{C}_{b,c}}(f,g)  \label{ineq}
\end{equation}%
for all $f,g\in \mathcal{C}_{b,c}.$ So by Banach's fixed point theorem for
metric spaces we deduce that the functional $\Phi_{T} :\mathcal{C}%
_{b,c}\rightarrow \mathcal{C}_{b,c}$ has a unique fixed point and, thus, the
recurrence equation (\ref{typrec}) has a unique solution.

In order to obtain the asymptotic complexity class of the solution to the
recurrence equation (\ref{typrec}), Schellekens introduced a special class
of functionals known as improvers.

Let $C\subseteq \mathcal{RT},$ a functional $\Phi :C\rightarrow C$ is called
an \textit{improver} with respect to a function $f\in C$ provided that $\Phi
^{n+1}(f)\leq \Phi ^{n}(f)$ for all $n\in \omega .$ Of course $\Phi ^{0}(f)=f$.

Observe that an improver is a functional which corresponds to a transformation on programs in such a way that the iterative applications of the transformation yield an improved, from a complexity point of view, program at each step of the iteration.

Note that when $\Phi $ is monotone, to show that $\Phi $ is an improver with
respect to $f\in C$, it suffices to verify that $%
\Phi (f)\leq f.$

Under these conditions the following result was stated in \cite{Sch95}:

\begin{theorem}
\label{Schell}A Divide and Conquer recurrence of the form (\ref{typrec}) has
a unique solution $f_{T}$ in $\mathcal{C}_{b,c}$. Moreover, if the monotone functional $\Phi_{T}$
associated to (\ref{typrec}), and given by (\ref{funct}), is an improver with respect to some function $%
g\in \mathcal{C}_{b,c},$ then the solution to the recurrence equation satisfies that $f_{T}\in \mathcal{O}(g)$.
\end{theorem}

In \cite{Sch95} Schellekens discussed the complexity class of Mergesort in order to illustrate the utility of Theorem \ref{Schell}. In the
particular case of Mergesort the recurrence equation (\ref{typrec}%
) in the average case behaviour is exactly%
\begin{equation}
T(n)=\left\{
\begin{array}{ll}
c & \text{\textrm{if }}n=1 \\
2T(\frac{n}{2})+\frac{n}{2} & \text{\textrm{if }}n\in \omega _{2}
\end{array}%
\right. .  \label{recmerge}
\end{equation}

Of course Theorem \ref{Schell} provides that the recurrence equation (\ref{recmerge}) has a unique solution $f_{T}\in \mathcal{C}_{2,c}$. In addition, Schellekens proved in  \cite%
{Sch95} that the functional $\Phi _{T}$ induced by the recurrence equation (%
\ref{recmerge}) is an improver with respect to a complexity function $g_{k}\in
\mathcal{C}_{b,c}$, with $k\in \mathbb{R^{+}}$ and $g_{k}(n)=kn\log _{2}(n)$ for all $n\in \omega _{2}$,
if and only if $\frac{1}{2}\leq k.$ Therefore, by Theorem \ref{Schell}, we
conclude that $f\in \mathcal{O}(g_{\frac{1}{2}}),$ i.e. Theorem \ref{Schell}
provides a formal proof of the well-known fact that the running time of
computing of Mergesort in the average case behaviour is in the asymptotic
complexity class of $n\log _{2}(n).$

When a program uses a recursion process to find the solution to a problem,
such a process is characterized by obtaining in each step of the computation
an approximation to the aforementioned solution which is better than the
approximations obtained in the preceding steps and, in addition, by
obtaining always the final approximation to the problem solution as the \textquotedblleft limit\textquotedblright\ of the computing process. A mathematical model to this sort of situations was developed by D.S. Scott
which is based on ideas from order theory and topology (see, for instance,
\cite{Scott}, \cite{Scott2} and \cite{Davey}). In particular, the order
represents some notion of information in such a way that each step of the
computation is identified with an element of the mathematical model which is
greater than (or equal to) the other ones associated with the preceding steps,
since each approximation gives more information about the final solution
than those computed before. The final output of the computational
process is seen as the limit of the successive approximations. Thus the
recursion processes are modeled as increasing sequences of elements of the
ordered set which converge to its least upper bound with respect to the given
topology. From an Information Theory point of view the Scott model needs to
distinguish two kind of elements: the so-called totally defined objects and
the partially defined objects. The former represent elements of the model
whose information content can not be stored in a computer and, therefore,
must be approximated. The partially defined objects match with those
elements of the model that can be stored in the computer and that
approximate the total ones.

In 1994 S.G. Matthews introduce the notion of partial metric as a
mathematical tool to model computational processes in the spirit of Scott
where a quantitative degree of the information content of the involved
elements is needed.

Let us recall some notions related to partial metrics in order to
explain the importance of this new metric concept.

According to \cite{Ma}, a partial metric on a nonempty set $X$ is a function
$p:X\times X\rightarrow \mathbb{R}^{+}$ such that for all $x,y,z\in
X: \medskip $

$\begin{array}{ll}
(i) & p(x,x)=p(x,y)=p(y,y)\Leftrightarrow x=y; \\
(ii) & p(x,x)\leq p(x,y); \\
(iii) & p(x,y)=p(y,x); \\
(iv) & p(x,y)\leq p(x,z)+p(z,y)-p(z,z).%
\end{array}%
\medskip $

A partial metric space is a pair $(X,p)$ such that $X$ is a nonempty set and
$p$ is a partial metric on $X.$

Each partial metric $p$ on $X$ generates a $T_{0}$ topology $\mathcal{T}(p)$
on $X$ which has as a base the family of open $p$-balls $\{B_{p}(x,%
\varepsilon ):x\in X,\varepsilon >0\}$, where $B_{p}(x,\varepsilon )=\{y\in
X:p(x,y)<p(x,x)+\varepsilon \}$ for all $x\in X$ and $\varepsilon >0.$ From
this fact it immediately follows that a sequence $(x_{n})_{n\in \mathbb{N}}$ in a partial
metric space $(X,p)$ converges to a point $x\in X$ with respect to $\mathcal{%
T}(p)\Leftrightarrow p(x,x)=\lim_{n\rightarrow \infty }p(x,x_{n}).$

As usual (\cite{Davey}) an order on a (nonempty) set $X$ is a reflexive, antisymmetric and transitive binary
relation $\leq $ on $X$, and a set equipped with an order is said to be an ordered set.

% such that, for all $x,y,z\in X,\medskip $

%$%
%\begin{array}{l}
%\mathrm{(i)}\text{ }x\leq x \\
%\mathrm{(ii)}\text{ }x\leq y\text{\textrm{\ and }}y\leq x\text{\textrm{\
%5implies }}x=y \\
%\mathrm{(iii)}\text{ }x\leq y\text{\textrm{\ and }}y\leq z\text{\textrm{\
%implies }}x\leq z%
%\end{array}%
%\begin{array}{l}
%\text{\textrm{(reflexivity)}} \\
%\text{\textrm{(antisymmetry)}} \\
%\text{\textrm{(transitivity).}}%
%\end{array}%
%\medskip $

The success of partial metrics lies in that every partial metric $p$ induces
an order $\leq _{p}$ on $X$ ($x\leq _{p}y\Leftrightarrow p(x,y)=p(x,x)$) in
such a way that increasing sequences of elements with respect to $\leq _{p}$
converge to its least upper bound with respect the partial metric topology $%
\mathcal{T}(p)$. Moreover, partial metrics can be used to distinguish
between objects with totally defined information content and objects with
partially defined information content in Scott's models. Specifically if $%
(X,p)$ is a partial metric space, then the numerical value $p(x,x)$ allows
us to quantify the amount of information contained in an object $x$. In
particular the smaller $p(x,x)$ the more defined $x$ is, being $x$ totally
defined if $p(x,x)=0$.

Motivated by the usefulness of partial metric spaces in Computer Science, it
seems natural to wonder if these kind of metric spaces are also useful in
complexity analysis in the spirit of Schellekens. However we show in the
following that the preceding question has at first a negative answer. To
this end, let us recall some additional and useful concepts about partial
metrics.

Following \cite{Ma}, a sequence $(x_{n})_{n\in \mathbb{N}}$ in a partial
metric space $(X,p)$ is called a Cauchy sequence if  $%
\lim_{n,m\rightarrow \infty }p(x_{n},x_{m})$ exists. A partial metric space $(X,p)$
is said to be complete if every Cauchy sequence $(x_{n})_{n\in \mathbb{N}}$
in $X$ converges, with respect to $\mathcal{T}(p),$ to a point $x\in X$ such
that $p(x,x)=\lim_{n,m\rightarrow \infty }p(x_{n},x_{m}).$

Inspired, in part, by the applications to program verification Matthews
extends Banach's fixed point theorem to the framework of partial metric
spaces, and he used it to formulate a suitable test for lazy data flow
deadlock in Khan's model of parallel computation (\cite{Ma2}). The
aforementioned new theorem can be stated as follows:

\begin{theorem}
\label{Matth}Let $f$ be a mapping of a complete partial metric space $(X,p)$
into itself such that there is $s\in \mathbb{R}^{+}$ with $0\leq s<1,$
satisfying
%\begin{equation}
$$p(f(x),f(y))\leq sp(x,y), $$
%\end{equation}%
for all $x,y\in X.$ Then $f$ has a unique fixed point. Moreover if $x\in X$
is the fixed point of $f,$ then $p(x,x)=0$.
\end{theorem}

According to \cite{Ma} and \cite{ORS} some correspondences between
quasi-metric and partial metric spaces can be stated.

\begin{proposition}
\label{weight}If $(X,p)$\ is a partial metric space, then the following statements hold:
\begin{itemize}
\item[(i)] The function $d_{p}:X\times X\rightarrow \mathbb{R}^{+}$ defined by $%
d_{p}(x,y)=p(x,y)-p(x,x)$ is a quasi-metric on $X$ such that $\mathcal{T}(p)=%
\mathcal{T}(d_{p}).$
\item[(ii)] \label{complete} The below
assertions are equivalent:
\begin{itemize}
\item[(1)] $(X,p)$ is complete

\item[(2)] $(X,d_{p})$ is bicomplete.
\end{itemize}
\end{itemize}
\end{proposition}

Following \cite{ORS} the set $\mathcal{C}$ can be endowed with a partial
metric $p_{\mathcal{C}}$ defined for all $f,g\in \mathcal{C}$ by
\begin{equation*}
p_{\mathcal{C}}(f,g)=\sum_{n=1}^{\infty }2^{-n}\max\left(\frac{1}{%
f(n)},\frac{1}{g(n)}\right) .
\end{equation*}%
Moreover it is not hard to see that the partial metric $p_{\mathcal{C}}$
induces, by Proposition \ref{weight}, the quasi-metric $d_{\mathcal{C}}.$ So
Proposition \ref{complete} guarantees that the partial metric space $(%
\mathcal{C},p_{\mathcal{C}})$ is complete. Also $(\mathcal{C}_{b,c},p_{%
\mathcal{C}}|_{\mathcal{C}_{b,c}})$ is complete, since $(\mathcal{C}%
_{b,c},d_{\mathcal{C}}|_{\mathcal{C}_{b,c}})$ is bicomplete.

Now suppose that there exists $s\in \mathbb{R}^{+}$ with $0\leq s<1$ such
that
\begin{equation*}
p_{\mathcal{C}}|_{\mathcal{C}_{b,c}}(\Phi_{T} (f),\Phi_{T} (g))\leq sp_{\mathcal{C}%
}|_{\mathcal{C}_{b,c}}(f,g).
\end{equation*}%
for all $f,g\in \mathcal{C}_{b,c}$.

We only consider the case of $0<s<1,$ because it is evident that the case $%
s=0$ gives a contradiction.

Take $f,g\in \mathcal{C}_{b,c}$ defined by $f(n)=2c$ and $g(n)=2(c+1)$ for
all $n\in \omega _{b}$. It is clear that%
\begin{equation*}
p_{\mathcal{C}}|_{\mathcal{C}_{b,c}}(\Phi_{T} (f),\Phi_{T} (g))=\frac{1}{2c}%
+\sum_{n=2}^{\infty }2^{-b^{n}}\max\left(\frac{1}{2ac+h(b^{n})},\frac{1}{%
2a(c+1)+h(b^{n})}\right).
\end{equation*} Moreover,

\begin{eqnarray}
p_{\mathcal{C}}|_{\mathcal{C}_{b,c}}(f,g)&=&\sum_{n=1}^{\infty
}2^{-b^{n}}\max\left( \frac{1}{f(b^{n})},\frac{1}{g(b^{n})}\right) \leq
\sum_{n=1}^{\infty }2^{-n}\max\left( \frac{1}{2c},\frac{1}{2(c+1)}\right) \nonumber \\
&=&\frac{1}{2c} \nonumber.
\end{eqnarray}%

Applying the hypothesis we obtain that

$$\frac{1}{2c} \leq p_{\mathcal{C}}|_{\mathcal{C}_{b,c}}(\Phi_{T} (f),\Phi_{T} (g))\leq sp_{\mathcal{C%
}}|_{\mathcal{C}_{b,c}}(f,g)\leq s\frac{1}{2c}.$$ As a result we deduce that $1\leq s<1,$ which is a contradiction.

Consequently Theorem \ref{Matth} can not be used to analyze the complexity
of the algorithms whose running time of computing is associated to a
recurrence equation (\ref{typrec}) when the partial metric\ $p_{\mathcal{C}%
}, $ instead the quasi-metric $d_{\mathcal{C}}$, is employed as a complexity
distance.

In the light of the preceding conclusion our purpose in this paper
is to demonstrate that partial metric spaces, and in particular the
partial metric fixed point theorem, can be used satisfactorily for the asymptotic complexity analysis of algorithms in the spirit of Schellekens' approach based on improvers. To achieve this goal
we focus our attention on a slight modification, that we have called
Baire partial quasi-metric, of the well-known Baire partial metric
on the set of all words over an alphabet. We show that such a
partial metric tool constitutes an appropriate implement to carry
out formally the asymptotic complexity analysis of algorithms whose
running time of computing leads to various types of recurrences
equations (not only the Divide and Conquer ones). In particular,
we validate the usefulness of our new results by means of applying them
to analyze the well-known asymptotic complexity of several
illustrative algorithms such as Quicksort, Mergesort and Largetwo.

The remainder of the paper is organized as follows: Section
\ref{BPM} is devoted to introduce the noted Baire partial metric
and, in addition, to prove that it allows to show the existence of solution to recurrence
equations that arise in a natural way in complexity analysis of algorithms. However, in the
same section we show that the Baire partial metric can not be
employed, in general, to obtain an asymptotic upper bound, and thus the complexity class, of the
running time of computing of an algorithm. Inspired by the last
fact, we introduce in Section \ref{BPQM} a new Baire partial
metric framework whose basis resides in the partial quasi-metric
approach introduced recently in \cite{Ku2}. We show that this new partial metric approach presents a relevant advantage with respect to the Schellekens one. More specifically, it is suitable to carry out the asymptotic complexity analysis of algorithms via fixed point methods without the need for assuming the convergence condition inherent to the definition of the complexity space in the Shellekens framework. Finally, the discussion
of the running time of the celebrated Quicksort, Mergesort and
Largetwo allows us, on one hand, to validate our new results and, on the other hand, to show the potential applicability of
the developed theory to complexity analysis in Computer Science.

\section{The Baire partial metric space and the asymptotic complexity analysis of
algorithms\label{BPM}}

With the aim of applying partial metric spaces to complexity
analysis of algorithms we recall the well-known Baire partial metric
space which was used by Matthews to model Kahn's parallel
computation processes (\cite{Kahn,Ma2}).

Let $\Sigma $ be a nonempty set (an alphabet). Denote by $\Sigma ^{\infty }$
the set of all finite and infinite sequences (words) over $\Sigma .$

For each $x\in \Sigma ^{\infty }$ we denote by $l(x)$ the length of $x$.
Hence $l(x)\in \lbrack 1,\infty ]$.

From now on, if $x\in \Sigma ^{\infty }$ with $l(x)=\infty $ we will write $%
x:=x_{1}x_{2}\ldots,$ and if $x\in \Sigma ^{\infty }$ with
$l(x)=n<\infty $ we will write $x:=x_{1}x_{2} \ldots x_{n}.$

Given $x,y\in \Sigma ^{\infty },$ denote by $l(x,y)$ the length of the
longest common prefix of $x$ and $y,$ i.e. $\emph{l}(x,y)=\sup \{n\in
\mathbb{N}:x_{k}=y_{k}$ whenever $k\leq n\}$ if $x$ and $y$ have a common
prefix, and $l(x,y)=0$ otherwise.

As usual the set $\Sigma ^{\infty }$ is equipped with the prefix order $%
\sqsubseteq ,$ which is defined by $x\sqsubseteq y\Leftrightarrow x$ is a
prefix of $y.$

Define on $\Sigma ^{\infty }\times \Sigma ^{\infty }$ the
nonnegative real valued function $p_{B}$ by
\begin{equation*}
p_{B}(x,y)=2^{-\emph{l}(x,y)}
\end{equation*}
for all $x,y\in \Sigma ^{\infty }$. Of course it is adopted the
convention that $2^{-\infty }=0$. It is not hard to see that the
pair $(\Sigma ^{\infty },p_{B})$ is a complete partial metric space
(\cite{ORS}), which is known as the Baire partial metric space.

The partial metric $p_{B}$ can be used to distinguish between
infinite words (objects with totally defined information content)
and finite words (objects with partially defined information
content) because $p_{B}(x,x)=0$ for some $x\in \Sigma ^{\infty
}\Leftrightarrow l(x)=\infty $. The fact that the value $p_{B}(x,x)$
allows to assign a degree of information content plays a crucial
role in Information Theory and, in addition, presents an advantage
with respect to the role played by the classical metrics extensively
used before, as for instance the well-known Baire metric (for a
fuller treatment of the classical Baire metric we refer the reader
to \cite{Calude,Pin}).

The Baire partial metric space $(\Sigma ^{\infty },p_{B})$ induces, by
Proposition \ref{weight}, the quasi-metric $d_{p_{B}}:\Sigma ^{\infty
}\times \Sigma ^{\infty }\rightarrow \mathbb{R}^{+}$ given by
\begin{equation*}
d_{p_{B}}(x,y)=2^{-\emph{l}(x,y)}-2^{-l(x)}
\end{equation*}
for all $x,y\in \Sigma ^{\infty }.$ Note that, similarly to the case of
complexity space (see Section \ref{Int}), we have $d_{p_{B}}(x,y)=0%
\Leftrightarrow x\sqsubseteq y.$

The next result provides that the Baire partial metric space is
suitable to show that the recurrence (\ref{typrec}) has a unique
solution.

\begin{theorem}
\label{Mattwords}Let $\Sigma =(0,\infty ]$ and $a,b\in \mathbb{N}$ with $a,b>1.$
Fix $z\in \Sigma ^{\infty }$ with $l(z)=\infty $ and $%
z_{k}\neq \infty $ for all $k\in \omega _{b}$ and $k\geq 2$. Let $\Sigma
_{b,c}^{\infty }$ be the subset of $\Sigma ^{\infty }$ given by $\Sigma
_{b,c}^{\infty }:=\{y\in \Sigma ^{\infty }: 2\leq l(y) \textrm{ and } y_{1}=c,y_{k}=\infty$ for
all $k\notin \omega _{b}$ with $2\leq k\leq l(y)$\}. Then the mapping
$\Theta _{a,b}^{z}:\Sigma _{b,c}^{\infty }\rightarrow \Sigma _{b,c}^{\infty }$ defined by $\Theta _{a,b}^{z}(x)=x_{\Theta _{a,b}^{z}}$,
where
\begin{equation*}
(x_{\Theta _{a,b}^{z}})_{k}:=\left\{
\begin{array}{ll}
c & \text{\textrm{if }}k=1 \\
\infty & \text{\textrm{if }}k\notin \omega _{b} \textrm{ and } 2\leq k\leq l(x)+1 \\
a\cdot x_{\frac{k}{b}}+z_{k} & \text{\textrm{if }}k\in \omega _{b} \text{\textrm{%
\ and }} \frac{k}{b}\leq l(x)
\end{array}%
\right.
\end{equation*}%
has a unique fixed point $v\in \Sigma _{b,c}^{\infty }$ with $%
l(v)=\infty .$
\end{theorem}

\begin{proof} It is easy to check that the set $\Sigma _{b,c}^{\infty }$ is closed in $%
(\Sigma ^{\infty },d_{p_{B}}^{s}).$ Hence $(\Sigma _{b,c}^{\infty },d_{p_{B}}^{s}|_{\Sigma _{b,c}^{\infty }})$ $\ $is
complete. By Proposition \ref{complete} the partial metric space
$(\Sigma _{b,c}^{\infty },p_{B}|_{\Sigma _{b}^{\infty }})$
is complete.

A straightforward computation shows that the mapping $\Theta
_{a,b}^{z}$
satisfies the inequality%
\begin{equation*}
p_{B}|_{\Sigma _{b}^{\infty }}(x_{\Theta _{a,b}^{z}},y_{\Theta _{a,b}^{z}})\leq \frac{1%
}{2}p_{B}|_{\Sigma _{b}^{\infty }}(x,y)
\end{equation*}%
for all $x,y\in \Sigma _{b,c}^{\infty }.$ Then, by Theorem
\ref{Matth}, we deduce that the mapping $\Theta _{a,b}^{z}$ has
unique fixed point $v\in \Sigma _{b,c}^{\infty }$ with
$p_{B}|_{\Sigma _{b,c}^{\infty }}(v,v)=0.$ Hence $l(x)=\infty .$
\end{proof}\medskip

Following the ideas introduced in Section \ref{Int} let us denote by $\mathcal{RT}_{b,c}$ the set
\begin{equation*}
\mathcal{RT}_{b,c}=\{f\in \mathcal{RT}:f(1)=c\text{ and }f(n)=\infty \text{
for all }n\notin \omega _{b}\text{ with }n>1\}.
\end{equation*} Of course if we take $z\in \Sigma ^{\infty }$ with $z_{k}=h(b^{k})$ for all $%
k\in \mathbb{N}$ in Theorem \ref{Mattwords}, where $h$ is
given as in the recurrence equation (\ref{typrec}), then the function $%
f_{v}\in \mathcal{RT}_{b,c}$ defined by $f_{v}(n)=v_{n}$ for all $%
n\in \mathbb{N}$, where $v$ is provided by Theorem \ref{Mattwords}, can be identified with the solution to the Divide and Conquer recurrence equation (\ref{typrec}) and, thus, with the running time of computing of Divide and Conquer
algorithms. However, to make a complete asymptotic complexity
analysis we must give the complexity class of the running time of
computing, i.e. the complexity class of $f_{v}$. Unfortunately the
presented framework is not able to allow us to get this objective.
The reason is given by the fact that two words $x,y\in \Sigma
^{\infty }$ with $l(x)=l(y)=\infty $ satisfy $x\sqsubseteq
y\Leftrightarrow x=y.$ So, according to the above technique, if $f,g$ represent the running time of computing of two algorithms ($f,g\in \mathcal{RT}$), and we
identify in a natural way those functions with the words $x^{f}, x^{g}\in \Sigma ^{\infty }$ ($%
\Sigma =(0,\infty ]$) such that $$x^{f}_{n}=f(n) \text{ and } x^{g}_{n}=g(n) \text{ for all } n\in \mathbb{N},$$ then we have
that
\begin{equation*}
x^{f}\sqsubseteq x^{g}\Leftrightarrow f(n)=g(n)\text{\textrm{\ for all }}%
n\in \mathbb{N} .
\end{equation*}
Therefore we can not obtain an asymptotic upper bound of the running time of
computing of an algorithm under analysis by means of the usual prefix order
defined on $\Sigma ^{\infty }.$ Note that in spite of the preceding disadvantage, the Baire partial metric framework provides the basis to develop a mathematical formalism for the complexity analysis of algorithms which does not rely on the technical ``convergence'' assumption incorporated into the definition of the complexity space $\mathcal{C}$, that is $$f\in \mathcal{C} \Leftrightarrow f\in\mathcal{RT} \text{ and }\sum_{n=1}^{+\infty
}2^{-n}\frac{1}{f(n)}<+\infty.$$ Of course, although every reasonable algorithm must hold the preceding convergence condition, this one is a little artificial and has the unique purpose of guaranteeing the finiteness of the value $d_{C}(f,g)$.

In the next subsection we propose, as a result of the preceding conclusion, a slight modification of the Baire partial metric which will allow to develop a suitable mathematical framework for asymptotic complexity analysis free from the convergence assumption.

\section{The Baire partial quasi-metric space and the asymptotic complexity analysis of
algorithms\label{BPQM}}

In the remainder of the paper we introduce a new metric tool on the set $%
\Sigma ^{\infty }$ of all words over an alphabet $\Sigma $ in such a way that a slight modification
of Theorem \ref{Matth} allows us to model satisfactorily the asymptotic
complexity analysis of algorithms via fixed point techniques in the spirit
of Schellekens.

\subsection{The Baire partial quasi-metric}
For our proposal we recall a few pertinent concepts.

In \cite{Ku2} H.A. K\"{u}nzi, H.A Pajooshesh and Schellekens have
introduced and studied the notion of partial quasi-metric. Roughly
speaking a partial quasi-metric is a partial metric which does not
satisfy the symmetry property. More specifically, a partial
quasi-metric\ on a nonempty set $X$ is a function $q:X\times
X\rightarrow \mathbb{R}^{+}$ such that for all $x,y,z\in X :\medskip
$

$\begin{array}{ll}
(i) & q(x,x)\leq q(x,y); \\
(ii) & q(x,x)\leq q(y,x); \\
(iii) & q(x,y)\leq q(x,z)+q(z,y)-q(z,z); \\
(iv) & q(x,x)=q(x,y)\text{\textrm{\ and }}q(y,y)=q(y,x)\Leftrightarrow x=y.%
\end{array}%
\medskip $

Observe that a partial metric on a set $X$ is a partial quasi-metric
satisfying in addition the condition: $$\begin{array}{ll}
                                         (v) & q(x,y)=q(y,x)\\
                                         \end{array}$$ for all
$x,y\in X.$

A partial quasi-metric space is a pair $(X,q)$ such that $X$ is a nonempty
set and $q$ is a partial quasi-metric on $X.$

Similarly to the case of partial metric spaces a partial quasi-metric $q$
generates a $T_{0}$-topology $\mathcal{T}(q)$ on $X$ which has as a base the
family of open $q$-balls $\{B_{q}(x,\varepsilon ):x\in X,$ $\varepsilon >0\}$%
, where $B_{q}(x,\varepsilon )=\{y\in X:q(x,y)<q(x,x)+\varepsilon \}$ for
all $x\in X$ and $\varepsilon >0.$

%It is clear that given a partial quasi-metric $q$ on $X$, the function $q^{+}
%$ defined on $X\times X$ by $q^{+}(x,y)=q(x,y)+q(y,x)$ is a partial metric
%on $X$.

If $q$ is a partial quasi-metric on $X$, then the function $d_{q}:X\times X\rightarrow \mathbb{R}^{+}$, defined by
$$d_{q}(x,y)=q(x,y)-q(x,x)$$ for all $x,y\in X$, is
a quasi-metric.

On account of \cite{Ku2}, a partial quasi-metric space $(X,q)$ is
said to be complete provided that the associated quasi-metric space
$(X,d_{q})$ is bicomplete. Moreover, in the same reference the
Matthews fixed point theorem (Theorem \ref{Matth} in Section
\ref{Int}) has been extended to the context of partial quasi-metric
spaces in the following way:

\begin{theorem}
\label{Kunzi}Let $f$ be a mapping from a complete partial quasi-metric space
$(X,q)$ into itself such that there is $s\in \mathbb{R}^{+}$ with $0\leq s<1,
$ satisfying
\begin{equation}
q(f(x),f(y))\leq sq(x,y),  \label{contrBPQM}
\end{equation}%
for all $x,y\in X.$ Then $f$ has a unique fixed point. Moreover if
$x\in X$ is the fixed point of $f,$ then $q(x,x)=0$.
\end{theorem}

The below result will be crucial in order to construct a partial
quasi-metric framework for the asymptotic complexity analysis.

\begin{proposition}\label{asymp}
\label{R}Under the conditions of Theorem \ref{Kunzi}, if there
exists $y\in X$ such that $d_{q}(f(y),y)=0$ then $d_{q}(x,y)=0.$
\end{proposition}

\begin{proof}Suppose for the purpose of contradiction that $d_{q}(x,y)\neq 0,$ where $x$ is the fixed point of $f$. Then $%
q(x,y)>0.$ Consequently the inequality (\ref{contrBPQM}) yields%
\begin{eqnarray*}
q(x,y) &\leq &q(x,f(x))+q(f(x),y)-q(f(x),f(x)) \\
&=& q(f(x),y)\\
&\leq &q(f(x),f(y))+q(f(y),y)-q(f(y),f(y)) \\
&=&q(f(x),f(y))+d_{q}(f(y),y) \\
&=&q(f(x),f(y))\leq sq(x,y).
\end{eqnarray*}%
It follows that $1\leq s<1$, which is a contradiction.
\end{proof}\medskip

Now we are able to construct a new metric tool on the set $\Sigma ^{\infty }.
$ Indeed, let $\Sigma $ be a nonempty alphabet endowed with an order $%
\preceq $. Then, given $x,y\in \Sigma ^{\infty },$ we will say that
$x$ is a subprefix of $y$, denoted by $x\sqsubseteq_{sp}y$, provided that there exists $n_{0}\in
\mathbb{N}$ with $n_{0}\leq l(x)$ such that $x_{k}\preceq y_{k}$ for
all $k\leq n_{0}.$ Obviously if $x\sqsubseteq y$, then $x\sqsubseteq_{sp}y$.

Note that, contrary to the case of the prefix, the subprefix notion does not induce an order relation on $\Sigma ^{\infty }$.

Let us denote by $\emph{l}_{\preceq }(x,y)=\sup \{n\in \mathbb{N} :x_{k}\preceq
y_{k}$ for all $k\leq n\}$ whenever $x\sqsubseteq_{sp}y$, and $\emph{l%
}_{\preceq }(x,y)=0$ otherwise. Note that
$$\emph{l}_{\preceq }(x,y)=\infty
\Leftrightarrow l(x)=l(y)=\infty \text{ and } x_{k}\preceq y_{k}
\text{ for all } k\in
\mathbb{N}.$$ Clearly $\emph{l}_{\preceq }(x,y)=\emph{l}(x,y)$ whenever $%
x\sqsubseteq y$, and $\emph{l}_{\preceq }(x,x)=l(x)$ for all $x\in
\Sigma ^{\infty }.$

In the light of the preceding definitions we have the following result.

\begin{proposition}\label{pqmB}
Let $\Sigma $ be an alphabet endowed with an order $\preceq $. Then
the pair $(\Sigma ^{\infty },q_{B})$ is a complete partial
quasi-metric space, where $q_{B}:\Sigma ^{\infty }\times \Sigma
^{\infty }\rightarrow \mathbb{R}^{+}$ is defined by
%\begin{equation*}
$q_{B}(x,y)=2^{-\emph{l}_{\preceq }(x,y)}$ for all $x,y\in \Sigma
^{\infty }$.
%\end{equation*}
\end{proposition}

\begin{proof}
Since $\emph{l}_{\preceq }(x,y)\leq l(x)$ and $\emph{l}_{\preceq }(y,x)\leq
l(x)$ for all $x,y\in \Sigma ^{\infty }$ we deduce immediately that $%
q_{B}(x,x)\leq q_{B}(x,y)$ and that $q_{B}(x,x)\leq q_{B}(y,x)$ for all $%
x,y\in \Sigma ^{\infty }.$

Next consider $x,y\in \Sigma ^{\infty }$ such that
$q_{B}(x,x)=q_{B}(x,y)$
and $q_{B}(y,y)=q_{B}(y,x).$ Then $l(x)=\emph{l}_{\preceq }(x,y)$ and $l(y)=%
\emph{l}_{\preceq }(y,x).$ It follows that $l(x)=l(y)$ and that $x=y.$

Now we show that $$q_{B}(x,y)\leq q_{B}(x,z)+q_{B}(z,y)-q_{B}(z,z)$$ for all $%
x,y,z\in \Sigma ^{\infty }.$

We assume that $\emph{l}_{\preceq }(x,z)>0$ and $\emph{l}_{\preceq }(z,y)>0,$
because otherwise it is clear that the preceding inequality holds. It
follows that $\emph{l}_{\preceq }(x,y)>0.$ Then it suffices to consider that
$\emph{l}_{\preceq }(x,y)\leq \min \left(\emph{l}_{\preceq }(x,z),\emph{l}%
_{\preceq }(z,y)\right).$ But this clearly forces that
$\emph{l}_{\preceq }(x,y)=\min \left(\emph{l}_{\preceq
}(x,z),\emph{l}_{\preceq }(z,y)\right).$
Consequently%
\begin{eqnarray*}
q_{B}(x,y) &=&2^{-\min \left(\emph{l}_{\preceq
}(x,z),\emph{l}_{\preceq }(z,y)\right)}\leq 2^{-\emph{l}_{\preceq
}(x,z)}+2^{-\emph{l}_{\preceq}(z,y)}-q_{B}(z,z)
\\
&=&q_{B}(x,z)+q_{B}(z,y)-q_{B}(z,z).
\end{eqnarray*}

Therefore we have shown that $q_{B}$ is a partial quasi-metric on $\Sigma
^{\infty }.$

Finally we show that the partial quasi-metric space $(\Sigma ^{\infty },q_{B})$
is complete. Indeed, let $(x_{n})_{n\in \mathbb{N}}$ be a Cauchy sequence in
$(\Sigma ^{\infty },d_{q_{B}}^{s}).$ Thus, given $\varepsilon >0,$ there
exists $n_{0}\in \mathbb{N}$ such that
\begin{equation*}
\max \left(\emph{2}^{-\emph{l}_{\preceq }(x_{n},x_{m})}-2^{-l(x_{n})},\emph{2}^{-%
\emph{l}_{\preceq }(x_{m},x_{n})}-2^{-l(x_{m})}\right)<\varepsilon
\end{equation*}%
for all $n,m\geq n_{0}.$

Since $(x_{n})_{k}=(x_{m})_{k}$ for all $k\leq \min
\left(\emph{l}_{\preceq }(x_{n},x_{m}),\emph{l}_{\preceq
}(x_{m},x_{n})\right)$ we have that
\begin{equation*}
p_{B}(x_{n},x_{m})=2^{-\min \left(\emph{l}_{\preceq }(x_{n},x_{m}),\emph{l}%
_{\preceq }(x_{m},x_{n})\right)}.
\end{equation*}
It follows that
\begin{equation*}
d_{p_{B}}^{s}(x_{n},x_{m})\leq 2\cdot 2^{-\min
\left(\emph{l}_{\preceq }(x_{n},x_{m}),\emph{l}_{\preceq
}(x_{m},x_{n})\right)}-2^{-l(x_{n})}-2^{-l(x_{m})}<3\varepsilon
\end{equation*}%
for all $m,n\geq n_{0}.$ Whence $\lim_{n,m\rightarrow \infty
}d_{p_{B}}^{s}(x_{n},x_{m})=0.$ Thus $(x_{n})_{n\in \mathbb{N}}$ is
a Cauchy sequence in $(\Sigma ^{\infty },d_{p_{B}}^{s}).$ Since
$(\Sigma ^{\infty },p_{B})$ is a complete partial metric space (see
Section \ref{BPM}) we have that there exists $x\in \Sigma ^{\infty }$ satisfying that
$$\lim_{n,m\rightarrow \infty }p_{B}(x_{n},x_{m})=p_{B}(x,x)=\lim_{n\rightarrow
\infty }p_{B}(x,x_{n}).$$ Whence we obtain that $\lim_{n\rightarrow
\infty }d_{p_{B}}(x,x_{n})=\lim_{n\rightarrow \infty
}d_{p_{B}}(x_{n},x)=0.$ Thus $$\lim_{n\rightarrow \infty
}2^{-\emph{l}_{\preceq }(x,x_{n})}=2^{-l(x)},$$ since
\begin{equation*}
0\leq 2^{-\emph{l}_{\preceq }(x,x_{n})}-2^{-l(x)}\leq 2^{-\emph{l}%
(x,x_{n})}-2^{-l(x)}<\varepsilon
\end{equation*}%
for all $n\in \mathbb{N}$ such that $n\geq n_{0}$.

Similar considerations apply to show that $\lim_{n\rightarrow \infty
}2^{-\emph{l}_{\preceq }(x_{n},x)}=2^{-l(x)}$ from the fact that
$\lim_{n\rightarrow \infty }d_{p_{B}}(x_{n},x)=0$. Hence
$\lim_{n\rightarrow \infty }d_{q_{B}}^{s}(x,x_{n})=0.$ Consequently
$(\Sigma ^{\infty },d_{q_{B}}^{s})$ is a bicomplete quasi-metric
space. Therefore $(\Sigma ^{\infty },q_{B})$ is a complete partial
quasi-metric space. The proof is complete.
\end{proof}\medskip

From now on the pair $(\Sigma ^{\infty },q_{B})$ will be called the
Baire partial quasi-metric space.

Notice that $$q_{B}(x,y)=0\Leftrightarrow l(x)=l(y)=\infty \text{
and } x_{k}\preceq y_{k} \text{ for all } k\in \mathbb{N} .$$ So the
Baire partial quasi-metric encodes, similarly to the case of the complexity quasi-metric $%
d_{\mathcal{C}}$ and the Baire partial metric $p_{B}$, the order
$\preceq $ on the set $\Sigma $ and the associated subprefix notion. Furthermore, we wish to emphasize that the Baire partial quasi-metric space remains valid to model all those processes which have been modeled by means of the Baire partial metric space (as, for instance, in program verification or in denotational semantics for programming languages). Nevertheless, the use of he Baire partial quasi-metric presents an advantage with respect to use of the partial metric one, and this advantage is given by its utility, contrarily to the Baire partial metric space, in complexity analysis.

\subsection{The asymptotic complexity analysis of algorithms via the Baire partial quasi-metric: Three examples}
We end the paper showing that the developed partial quasi-metric theory is useful to analyze the
asymptotic complexity of algorithms. Furthermore, we validate our results retrieving as a particular case the well-known asymptotic complexity of Mergesort, Quicksort and Largetwo. To this end, let us recall that, as set out in Section \ref{Int}, when discussing the running time of computing of Divide and Conquer algorithms usually one has to solve recurrence equations of the form
\begin{equation}
T(n)=\left\{
\begin{array}{ll}
c & \text{\textrm{if }}n=1 \\
aT(\frac{n}{b})+h(n) & \text{\textrm{if }}n\in \omega_{b}%
\end{array}%
\right. ,  \label{eqDC}
\end{equation}%
where $a,b\in \mathbb{N}$ with $a,b>1$, $c\in \mathbb{R^{+}}$ with $c>0$ and $h\in\mathcal{RT}$ such that $0<h(n)<\infty $ for all $n\in \mathbb{N}$.

The asymptotic complexity of those algorithms whose running time of computing is given by the solution of a recurrence equation (\ref{eqDC}) can be analyzed via the next result which is a Baire partial quasi-metric space version of Theorem \ref{Schell} in Section \ref{Int}.

\begin{theorem}\label{UsefulTh1}
Let $\Sigma =(0,\infty ]$ and $a,b\in \mathbb{N}$ with $a,b>1.$
Fix $z\in \Sigma ^{\infty }$ with $l(z)=\infty $ and $%
z_{k}\neq \infty $ for all $k\in \omega _{b}$ and $k\geq 2$. Let $\Sigma
_{b,c}^{\infty }$ be the subset of $\Sigma ^{\infty }$ given by
$\Sigma_{b,c}^{\infty }:=\{y\in \Sigma ^{\infty }: 2\leq l(y) \textrm{ and } y_{1}=c,y_{k}=\infty \text{ for
all } k\notin \omega _{b} \text{ with } 2\leq k\leq l(y)\}$. Then the mapping
$\Theta _{a,b}^{z}:\Sigma _{\ast ,b}^{\infty }\rightarrow \Sigma _{b,c}^{\infty }$ defined by $\Theta _{a,b}^{z}(x)=x_{\Theta _{a,b}^{z}}$,
where
\begin{equation*}
(x_{\Theta _{a,b}^{z}})_{k}:=\left\{
\begin{array}{ll}
c & \text{\textrm{if }}k=1 \\
\infty & \text{\textrm{if }}k\notin \omega _{b} \textrm{ and } 2\leq k\leq l(x)+1 \\
a\cdot x_{\frac{k}{b}}+z_{k} & \text{\textrm{if }}k\in \omega _{b} \text{\textrm{%
\ and }} \frac{k}{b}\leq l(x)
\end{array}%
\right.
\end{equation*}%
has a unique fixed point $v\in \Sigma _{b,c}^{\infty }$ with $%
l(v)=\infty .$ Moreover if $u\in \Sigma _{b,c}^{\infty }$ such that $\Theta _{a,b}^{z}(u)\sqsubseteq_{sp}u$ then $v\sqsubseteq_{sp}u$.
\end{theorem}

\begin{proof} First of all we note that the subset  $\Sigma _{b,c}^{\infty }$ is closed in $(\Sigma ^{\infty },d_{q_{B}}^{s})$
and, thus, the pair $(\Sigma _{b,c}^{\infty },q_{B}|_{\Sigma _{b,c}^{\infty }})$ is a complete partial quasi-metric space. Moreover, the
mapping $\Theta _{a,b}^{z}:\Sigma _{b,c}^{\infty }\rightarrow \Sigma
_{b,c}^{\infty }$ holds the inequality (\ref{contrBPQM}) in Theorem \ref{Kunzi} with $s=\frac{1}{2}.$ So the
aforesaid theorem guarantees that $\Theta _{a,b}^{z}$ has a unique fixed
point $v\in \Sigma _{b,c}^{\infty }$ with $q_{B}(v,v)=0.$ Hence $l(v)=\infty .$

Now assume the existence of $u\in \Sigma _{b,c}^{\infty }$ such that $\Theta _{a,b}^{z}(u)\sqsubseteq_{sp}u$. It follows, by construction of $\Theta_{a,b}^{z}$, that $l(u)=\infty$. Furthermore, it is clear that $d_{q_{B}}(\Theta _{a,b}^{z}(u),u)=0.$ Hence, by Proposition \ref{asymp}, we have that $d_{q_{B}}(v,u)=0.$ Whence we deduce that $v\sqsubseteq_{sp}u$.\end{proof}$\medskip$

Similarly to Section \ref{Int}, we denote by  $\Phi_{T}$ the mapping $\Phi_{T}:\mathcal{RT}_{b,c}\rightarrow \mathcal{RT}_{b,c}$ given by
\begin{equation*}
\Phi _{T}(f)(n)=\left\{
\begin{array}{ll}
c & \text{\textrm{if }}n=1 \\
\infty  & \text{\textrm{if }}n\notin \omega _{b} \textrm{ and } n>1\\
af(\frac{n}{b})+h(n) & \text{\textrm{otherwise}}%
\end{array}%
\right. .
\end{equation*}

\begin{corollary}\label{CorUseful1}A Divide and Conquer recurrence of the form (\ref{eqDC}) has
a unique solution $f_{T}\in \mathcal{RT}_{b,c}$. Moreover if there exists $g\in \mathcal{RT}_{b,c}$ such that $\Phi_{T}$ is an improver with respect to $g$, then $f_{T}\in \mathcal{O}(g)$.
\end{corollary}

\begin{proof}Let $v\in \Sigma_{b,c}^{\infty }$ be the fixed point of the mapping
$\Theta _{a,b}^{z}$ ensured by Theorem \ref{UsefulTh1}. Define $f_{v}\in \mathcal{RT}$
following the same arguments as in Section \ref{BPM}. We immediately obtain that $f_{v}\in \mathcal{RT}_{b,c}$ is the unique solution $f_{T}$ to the
recurrence equation (\ref{eqDC}). So $f_{v}$ can be
identified with the running time of computing of a Divide and Conquer
algorithm. In addition if $\Phi_{T}$ is an improver with respect to $g\in \mathcal{RT}_{b,c}$, then we can identify such a
complexity function with a word $y^{g}\in \Sigma _{b,c}^{\infty }$,
defined by $y_{k}^{g}=g(k)$ for all $k\in \omega _{b},$ such that $\Theta _{a,b}^{z}(y^{g})\sqsubseteq_{sp}y^{g}$. It follows, by Theorem \ref{UsefulTh1}, that $v\sqsubseteq_{sp}y^{g}$, that is
$f_{v}\in \mathcal{O}(f_{y^{g}}).$ Since $%
f_{y^{g}}=g$ we have obtained that $f_{v}\in \mathcal{O}(g).$ \end{proof} \medskip

Typical examples of algorithms whose running time of computing is the solution to a recurrence equation of kind (\ref{eqDC}) are Mergesort and   (in the best case behaviour). \bigskip

In the case of Mergesort the recurrence equation (\ref{eqDC}) in the worst case behaviour is the following:

\begin{equation}
T(n)=\left\{
\begin{array}{ll}
c & \text{\textrm{if }}n=1 \\
2T(\frac{n}{2})+n-1 & \text{\textrm{if }}n\in \omega _{2}
\end{array}%
\right. ,  \label{rec4}
\end{equation}%
and in the best and the average case behaviour is exactly the next one:%
\begin{equation}
T(n)=\left\{
\begin{array}{ll}
c & \text{\textrm{if }}n=1 \\
2T(\frac{n}{2})+\frac{n}{2} & \text{\textrm{if }}n\in \omega _{2}%
\end{array}%
\right. .  \label{rec5}
\end{equation}%
\bigskip

When Quicksort is considered the recurrence equation (\ref{eqDC}) associated to the running time of
computing in the best case behaviour is exactly the following one:%
\begin{equation}
T(n)=\left\{
\begin{array}{ll}
c & \text{\textrm{if }}n=1 \\
2T(\frac{n}{2})+dn & \text{\textrm{if }}n\in \omega _{2}
\end{array}%
\right. , \label{rec3}
\end{equation}%
where $d\in \mathbb{R^{+}}$ wiht $d>0$.\medskip

As an immediate consequence of Theorem \ref{UsefulTh1} and Corollary \ref{CorUseful1} we obtain the following well-known results which ratify, in part, the proposed theory.

\begin{corollary}Let $r\in \mathbb{R^{+}}$ with $r>0$. Define the mapping $g^{r}_{log_{2}}\in \mathcal{RT}_{b,c}$ by

\begin{equation*}
g_{\log _{2}}^{r}(n)=\left\{
\begin{array}{ll}
c & n=1 \\
\infty & n\notin \omega _{2} \textrm{ and } n>1 \\
rn\log _{2}n & \text{\textrm{otherwise}}%
\end{array}%
\right. .
\end{equation*}

Then the running time of computing of
\begin{itemize}
\item[1)] Mergesort in the worst case behaviour is in the complexity class $\mathcal{O}(g^{1}_{log_{2}})$.
\item[2)] Mergesort in the best and the average case bahviour is in the complexity class $\mathcal{O}(g^{\frac{1}{2}}_{log_{2}})$.
\item[3)] Quicksort in the best case bahaviour is in the complexity class $\mathcal{O}(g^{d}_{log_{2}})$.
\end{itemize}
\end{corollary}

Next we show that the techniques based on the Baire partial quasi-metric space are
applicable to a more general class of algorithms than those whose running
time of computing can be associated with a solution to the recurrence
equation (\ref{eqDC}).

In spite of seeming natural that the complexity analysis of Divide and Conquer algorithms always leads to recurrence equations of type (\ref{eqDC}), sometimes these kind of recursive algorithms yield recurrence equations that differ from (\ref{eqDC}). A well-known example of this type of situation is provided by Quicksort. In the worst case behaviour the recurrence equation obtained for Quicksort is given exactly as follows:

\begin{equation}
T(n)=\left\{
\begin{array}{ll}
c & \text{\textrm{if }}n=1 \\
T(n-1)+jn & \text{\textrm{if }}n\geq 2%
\end{array}%
\right. ,  \label{rec1}
\end{equation}%
where $j\in \mathbb{R^{+}}$ with $j>0$. Observe that in this case it is not necessary to restrict the input size of the data to the set $%
\omega _{b}$ for some $b\in \mathbb{N}$ with $b>1.$

Another example of algorithms, in this case a non recursive algorithm, whose complexity analysis leads to a
recurrence equation different from (\ref{eqDC}) is the well-known Largetwo. This finds the two largest entries in one-dimensional array of
size $n\in \mathbb{N}$ with $n>1$ (for a deeper discussion see \cite%
{Cull}). The running time of computing of Largetwo in the average case
behaviour can be associated with the solution to the recurrence equation
given as follows:

\begin{equation}
T(n)=\left\{
\begin{array}{ll}
c & \text{\textrm{if }}n=1 \\
T(n-1)+2-\frac{1}{n} & \text{\textrm{if }}n\geq 2%
\end{array}%
\right. ,  \label{rec2}
\end{equation}%
where $c$ is, again, the time taken by the algorithm in the base case, i.e.
when the input data is a one-diemensional array with only one element or the
array does not contain input data. Notice that Largetwo needs inputs data
with size at least $2$.

Of course the recurrence equations that yield the running time of computing
of the above aforesaid algorithms can be considered as particular cases
of the following general one:

\begin{equation}
T(n)=\left\{
\begin{array}{ll}
c & \text{\textrm{if }}n=1 \\
T(n-1)+h(n) & \text{\textrm{if }}n\geq 2%
\end{array}%
\right. ,  \label{rec0}
\end{equation}%
where $c\in \mathbb{R^{+}}$ with $c>0$ and $h\in \mathcal{RT}$ such that $0<h(n)<\infty $ for all $n\in \mathbb{N}$.\medskip

Setting $$\mathcal{RT}_{c}=\{f\in \mathcal{RT}: f(1)=c\}$$ and defining $\Gamma_{T} :\mathcal{RT}_{c}\rightarrow \mathcal{RT}_{c}$ by

\begin{equation}
\Gamma_{T}(f)(n)=\left\{
\begin{array}{ll}
c & \text{\textrm{if }}n=1 \\
f(n-1)+h(n) & \text{\textrm{if }}n\geq 2%
\end{array}%
\right. , \label{funct1}
\end{equation} we have that similar considerations to those given in the proof of Theorem \ref{UsefulTh1} apply to next one, which gives a method based on $\mathcal{RT}_{c}$ (Corollary \ref{CorUseful2}) to describe the complexity of those algorithms whose running time of computing satisfies the recurrence equation (\ref{rec0}).

\begin{theorem}
\label{Usefulth2}Let $\Sigma =(0,\infty ].$ Fix $z\in \Sigma ^{\infty }$
with $l(z)=\infty $ and $z_{k}\neq \infty $ for all $k\in \mathbb{N}$ and $%
k\geq 2.$ Let $\Sigma _{c}^{\infty }$ be the subset of $\Sigma ^{\infty
} $ given by $\Sigma _{c}^{\infty }:=\{y\in \Sigma ^{\infty
}:2\leq l(y) \text{ and } y_{1}=c\}.$ Then the mapping $\Psi
^{z}:\Sigma _{c}^{\infty }\rightarrow \Sigma _{c}^{\infty }$ defined
by $\Psi ^{z}(x)=x_{\Psi^{z}}$, where
\begin{equation*}
(x_{\Psi ^{z}})_{k}:=\left\{
\begin{array}{ll}
c & \text{\textrm{if }}k=1 \\
x_{k-1}+z_{k} & \text{\textrm{if }}2\leq k\leq l(x)%
\end{array}%
\right. ,
\end{equation*}%
has a unique fixed point $v\in \Sigma _{c}^{\infty }$ with $l(v)=\infty
. $ Moreover if $u\in \Sigma _{c}^{\infty }$ such that $\Psi^{z}(u)\sqsubseteq_{sp}u$ then $v\sqsubseteq_{sp}u$.
\end{theorem}

\begin{corollary}\label{CorUseful2}A recurrence of the form (\ref{rec0}) has
a unique solution $f_{T}\in \mathcal{RT}_{c}$. Moreover if there exists $g\in \mathcal{RT}_{c}$ such that $\Gamma_{T}$ is an improver with respect to $g$, then $f_{T}\in \mathcal{O}(g)$.
\end{corollary}

From Theorem \ref{Usefulth2} and Corollary \ref{CorUseful2} we immediately deduce the next well-known results.

\begin{corollary}Let $d,r\in \mathbb{R^{+}}$ with $d,r>0$. Then the following assertions hold:
\begin{itemize}
\item[1)] The running time of computing of Quicksort in the worst case behaviour is in the the complexity class $\mathcal{O}(g_{k})$, where
$k=\max\left(\frac{c}{4}+\frac{j}{2},\frac{3j}{5} \right)$ and
$$
g_{r}(n)=\left\{
\begin{array}{ll}
c & \text{\textrm{if }}n=1 \\
rn^{2} & \text{\textrm{if }}n\geq 2%
\end{array}%
\right. .
$$

\item[2)] The running time of computing of the Largetwo in the average case behaviour is in the the complexity class $\mathcal{O}(g_{k})$, where $k=\max\left(\frac{2c+3}{2+2d},1\right)$ and
$$
g_{r}(n)=\left\{
\begin{array}{ll}
c & \text{\textrm{if }}n=1 \\
r(2(n-1)-\log _{2}n+d) & \text{\textrm{if }}n\geq 2%
\end{array}%
\right.
$$
\end{itemize}
\end{corollary}

\section{Conclusions}
Partial metric spaces play a distinguished role in Computer Science.
Motivated by this fact we have discussed their usefulness for analyzing the
complexity of algorithms. In particular we have shown that the concept of
partial quasi-metric space, directly related to the partial metric one, is
appropriate to carry out, without convergence assumptions, the asymptotic complexity analysis of algorithms in
the spirit of Schellekens. In particular  we have constructed the Baire
partial quasi-metric from a new prefix notion between words over an
alphabet, and we have applied the new partial metric structure, via fixed
point arguments, to discuss the asymptotic complexity of algorithms whose running time of computing is typically given by a
recurrence equation. The running time of computing of Mergesort, Quicksort and Largetwo has been analyzed as specific examples in order to validate the developed theory.

\section{Acknowledgements}
The second author acknowledges the support of the Science Foundation
Ireland, SFI Principal Investigator Grant 07/IN.1/I977, and wishes to thank the Universidad de las Islas Baleares, where the paper was written during his research visit in September (2009), for financial support and hospitality. The third author
acknowledges the support of the Spanish Ministry of Science and Innovation,
and FEDER, grant MTM2009-12872-C02-01.


\begin{thebibliography}{99}
\bibitem{Bratley} G. Brassard, P. Bratley, \textit{Algorithms: Theory and
Practice}, Prentice Hall, New Jersey (1988).

\bibitem{Calude} C.S. Calude, S. Marcus, L. Staiger, \textit{A topological
characterization of random sequences, }Inform. Process. Lett. \textbf{88}, 245-250 (2003).

\bibitem{Cormen} T.H. Cormen, C.E. Leiserson, R.L. Rivest, \textit{%
Introduction to Algorithms, }MIT Press, New York (1990).

\bibitem{Cull} P. Cull, M. Flahive, R. Robson, \textit{Difference equations:
From rabbits to chaos, }Springer, New York (2005).

\bibitem{Davey} B.A. Davey, H.A. Priestley, \textit{Introduction to Lattices
and Order, }Cambridge University Press, Cambridge (1990).

%\bibitem{Flajolet} P. Flajolet, M.J. Golin, \textit{Mellin transforms and
%asymptotics: Mergesort recurrence, }Acta informatica \textbf{31},
%673-696 (1994).

\bibitem{Kahn} G. Kahn, \textit{The semantics of a simple language for
parallel processing.} In: Proc. of the IFIP
Congress Stockholm,  pp. 471-475. Elsevier and North-Holland, Amsterdam (1974).

\bibitem{Ku2} H.P.A. K\"{u}nzi, H. Pajooshesh, M.P. Schellekens, \textit{%
Partial quasi-metrics, }Theoret. Comput. Sci. \textbf{365}, 237-246 (2006).

\bibitem{Ku} H.P.A. K\"{u}nzi, \textit{Nonsymmetric distances and their
associated topologies: About the origins of basic ideas in the area of
asymmetric topology. } In: C.E. Aull and R. Lowen (eds.) Handbook of the History of General Topology vol. 3, pp. 853-968. Kluwer, Dordrecht (2001).

\bibitem{Ma} S.G. Matthews, \textit{Partial metric topology,} Ann. New York Acad.
Sci. \textbf{728}, 183-197 (1994).

\bibitem{Ma2} S.G. Matthews, \textit{An extensional treatment of lazy data
flow deadlock, }Theoret. Comput. Sci. \textbf{151}, 195-205 (1995).\textit{\ }

\bibitem{ORS} S. Oltra, S. Romaguera, E.A. S\'{a}nchez-P\'{e}rez, \textit{%
Bicompleting weightable quasi-metric spaces and partial metric spaces, } Rend. Circolo Mat. Palermo \textbf{51}, 151-162 (2002).

\bibitem{Pin} D. Perrin, J.E. Pin, \textit{Infinite words: Automata,
Semigroups, Logic and Games, }Pure and Appl. Math. Series, vol. 141,
Elsevier Acad. Press, Amsterdam (2004).

\bibitem{RomSch} S. Romaguera, M.P. Schellekens, \textit{Quasi-metric
properties of complexity spaces, }Topology Appl. \textbf{98}, 311-322 (1999).

\bibitem{Sch95} M. Schellekens,\textit{\ The Smyth completion: a common
foundation for denonational semantics and complexity analysis,} Electronic Notes in Theoret. Comput. Sci. \textbf{1},
211-232 (1995).

\bibitem{Scott} D. S. Scott, \textit{Outline of a mathematical theory of
computation. }In: Proc. 4th Annual Princeton Conference on Information
Sciences and Systems, pp. 169-176 (1970).

\bibitem{Scott2} D. Scott, \textit{Lattice theory, data types and semantics.
}In: Proc. Courant Computer Science Symposium on Formal Semantic of
Programming Languages, pp. 66-106, Prentice-Hall, Englewood Cliffs (1972).
\end{thebibliography}
\end{document}